\documentclass[letterpaper, 10 pt, conference]{ieeeconf}
\IEEEoverridecommandlockouts
\usepackage{mathrsfs}
\usepackage{amssymb}
\usepackage{amsmath,bm}
\usepackage{bbold}

\usepackage{mathrsfs, textcomp}
\usepackage{graphicx, caption}

\usepackage{subcaption}

\usepackage{enumerate}
\usepackage{eurosym}
\usepackage{amssymb}
\usepackage{amsmath,bbm}
\usepackage{amsfonts}
\usepackage{epstopdf}
\usepackage{epsf}
\usepackage{psfrag}
\usepackage{graphics}
\usepackage{color} 
\usepackage{cite, booktabs}
\usepackage{array,multirow,pbox}
\usepackage{tikz}
\usetikzlibrary{patterns,angles,quotes,arrows}

\newtheorem{theorem}{Theorem}
\newtheorem{proposition}{Proposition}
\newtheorem{remark}{Remark}
\newtheorem{definition}{Definition}

\newtheorem{problem}{Problem}

\graphicspath{{../Figures/}}

\begin{document}
%


\title{\vspace{0.5in}Distributed Transient Safety Verification via Robust Control Invariant Sets: A Microgrid Application}

\author{Jean-Baptiste Bouvier$^{1,2}$, Sai Pushpak Nandanoori$^{1}$, Melkior Ornik$^{2}$ and Soumya Kundu$^{1}$
\thanks{$^{1}$Jean-Baptiste Bouvier, Sai Pushpak Nandanoori and Soumya Kundu are with the Pacific Northwest National Laboratory, Richland, WA, USA. 
        {\tt\small \{saipushpak.n,\,soumya.kundu\}@pnnl.gov}}%
\thanks{$^{2}$Jean-Baptiste Bouvier and Melkior Ornik are with the Department of Aerospace Engineering and the Coordinated Science Laboratory at the University of Illinois Urbana-Champaign, Urbana, IL, USA. \break
        {\tt\small \{bouvier3, mornik\}@illinois.edu }}%
}


\maketitle

\begin{abstract}
Modern safety-critical energy infrastructures are increasingly operated in a hierarchical and modular control framework which allows for limited data exchange between the modules. In this context, it is important for each module to synthesize and communicate constraints on the values of exchanged information in order to assure system-wide safety. To ensure transient safety in inverter-based microgrids, we develop a set invariance-based distributed safety verification algorithm for each inverter module. Applying Nagumo's invariance condition, we construct a robust polynomial optimization problem to jointly search for safety-admissible set of control set-points and design parameters, under allowable disturbances from neighbors. We use sum-of-squares (SOS) programming to solve the verification problem and we perform numerical simulations using grid-forming inverters to illustrate the algorithm. 
\end{abstract}


%
\IEEEpeerreviewmaketitle

\section{Introduction}

The massive failure of the Texas electrical grid in February 2021 \cite{Texas} gave global coverage to the issue of power network resilience.
During these extreme events, time and resources are of essence for the grid operator to assess the situation and take appropriate actions to maintain the operating state of the power network. Hence, to an extent it is imperative on the grid operator to be prepared for extreme events. Microgrids, both grid-connected and stand-alone, have shown promise to enhance resilience and reliability by paving a way of coordinating multiple distributed energy resources (DERs) as a locally operated single controllable entity \cite{farrokhabadi2019microgrid,xu2016microgrids}. However, ensuring operational stability, safety, and reliability of any power network involves a complex multi-timescales problem, spanning sub-seconds to minutes and hours. Traditional power system control operations were largely structured around a temporal decoupling which allows slower-timescale operations (e.g., optimal dispatch) need not directly take into account faster-timescale constraints, and \textit{vice versa}. However, with the emergence of inverter-based DERs and the associated changes in power systems dynamics (e.g., reducing inertia), the timescales separation is expected to continue to shrink \cite{taylor2016power}. The droop-controlled inverter based microgrids have a lower inertia than conventional generators, which allows large variations of the voltage and frequency of each inverters \cite{xu2016microgrids, maulik2018stability}. The fluctuations happen during the transient evolution occurring as a result of a fault or due to the transitions between the power set-points, and can lead to violation of safety constraints \cite{schiffer2014conditions}. It is imperative to develop a mechanism to inform slower-timescale operations (e.g., optimal dispatch of power set-points) of the constraints arising from faster-timescale (transient) dynamics.

Many recent efforts have addressed this need via, for example, stability and security-constrained optimization \cite{barklund2008energy,xu2016microgrids,maulik2018stability}, identification of local and distributed parametric stability conditions \cite{schiffer2014conditions,kundu2019identifying,nandanoori2020distributed}, etc. The distributed identification of stability conditions are particularly interesting since these fit well into a multi-ownership models of microgrid resources, and facilitate hierarchical and plug-and-play operations \cite{guerrero2010hierarchical}. However, most of the related literature, as above, only focus on the stability which concerns with the convergence of power system trajectories (close) to its normal operating point after a disturbance. The concept of safety, on the other hand, relates to avoiding critical operational limits (e.g., on voltages and frequencies), even under large disturbances. Safety is closely tied to resilience, since often in cyber-physical adversarial scenarios, an immediate priority is to contain the system trajectories within some acceptable set, rather than ensuring return to normality.

Safety-constrained control techniques are gaining recent attention in the power systems community \cite{kundu2019distributed,chen2019compositional,zhang2020distributed,kundu2020transient}. Model-predictive control \cite{mayne2000constrained} remains one of the most commonly used methods for enforcing dynamic constraints over some prediction horizon. However, it suffers from certain limitations in the context of power system dynamics, related to, for example, nonlinearity and associated complexity of the dynamics, information disparity due to communication overheads and/or privacy concerns, and computational burden, especially for longer prediction horizon \cite{zhang2020distributed}. To circumvent these issues, distributed safety verification and control methods based on robust forward set-invariance principles have been proposed in \cite{kundu2019distributed,chen2019compositional,zhang2020distributed,kundu2020transient}. However, these prior works rely on the existence or the construction of parametric Lyapunov functions and/or barrier functions, thereby often incurring prohibitive computational costs and resulting in conservative safety certificates. The work in \cite{kundu2019distributed}, for example, proposes a sum-of-squares (SOS) programming based computation algorithm for distributed safety certificates as a super-level set of barrier functions. However, such computational methods result in conservative estimates of the safety-guaranteed set (e.g., Fig.\,1 in \cite{kundu2019distributed}), and typically do not scale well.

In this work, we consider a hierarchical and modular microgrid control architecture, \cite{guerrero2010hierarchical, kundu2020transient,almassalkhi2020hierarchical}, which allows system-level dispatch of power set-points to inverter-based resources, accommodating only limited data exchange between the (neighboring) inverter modules. Our objective is to design bounds on dispatched control set-points at the inverter buses, that guarantee, in a distributed sense, the safe excursions of local voltage and frequency within the specified limits while tolerating uncertainty in the neighboring buses. Specifically, the work presented in this paper relies on the \textit{Nagumo's theorem} \cite{Nagumo,invariance} to build an efficient method for distributed and robust safety verification, without requiring existence or construction of barrier functions. Thus, the proposed approach requires less computation, and relaxes the conservativeness of the barrier-certified safe sets by directly accommodating the original safety specifications.
The explicit reference governor \cite{nicotra2018explicit} relies on the same concept but we focus on establishing safe bounds on control inputs, instead of deriving a specific control law, which is the role of the grid coordinator \cite{kundu2020transient}.

The main contributions of this article are threefold. 
Firstly, we determine the maximal interval (bounds) of the dispatched control set-points guaranteeing the safety of droop-controlled inverters.
Secondly, we establish the monotonic relationship between this interval of safety admissible control set-points and the droop coefficients.
Finally, we calculate efficiently the maximal droop coefficient for which these safety admissible controls exist. As such, this paper provides novel design guidelines for the droop control parameters, extending the literature on stability-informed droop settings, e.g., \cite{schiffer2014conditions,kundu2019identifying,nandanoori2020distributed}, and the references therein. We use SOS programming to solve the safety verification problems, and illustrate the algorithm via numerical simulations. %
%
%
%
%
%
%
%
%
The remainder of this article is structured as follows. Section~\ref{sec:prelim} provides a background of the relevant theory. Section~\ref{sec:problem} presents our microgrid problem of interest. Section~\ref{sec:theory} explains the theoretical approach which we illustrate on a numerical example in Section~\ref{sec:example}. We conclude this article with Section~\ref{sec:conclusion}.

\section{Preliminaries}\label{sec:prelim}

\subsection{Invariant Sets}\label{subsec:invariant}

Consider a nonlinear dynamical system of the form
\begin{equation}\label{eq:f(x)}
   \dot x(t) = f\big( x(t) \big), \qquad x \in \mathbb{R}^n,
\end{equation}
with $f$ a Lipschitz continuous function. The objective of our work is to identify safe sets that the state $x$ cannot leave. Such sets are called \textit{invariant} (or \textit{positively invariant}); we define them as in \cite{invariance}.

\begin{definition}\label{def:invariance}
    A set $S$ is \emph{invariant} by the dynamics \eqref{eq:f(x)} if $x(0) \in S$ yields $x(t) \in S$ for all $t \geq 0$.
\end{definition}

To characterize invariant sets we will be using a theorem first established by Nagumo \cite{Nagumo} and then independently rediscovered by Brezis \cite{Brezis}. We state here a more modern formulation of this result from \cite{invariance}.

\begin{theorem}[Nagumo 1942]\label{thm:BonyBrezis}
    A closed set $S$ is invariant by the dynamics \eqref{eq:f(x)} if and only if for all $x \in S$, $f(x) \in \mathcal{C}(x)$, with $\mathcal{C}(x)$ the Bouligand tangent cone to $S$ at $x$.
\end{theorem}

A full definition of the Bouligand tangent cone is given in \cite{invariance}, but we will be studying sets $S$ where the Bouligand tangent cone is $\mathbb{R}^+$ or $\mathbb{R}^-$. The geometrical interpretation of Nagumo's theorem is that $f$ pointing inside $S$ on its boundary prevents trajectories from leaving $S$.

        
        
        

\subsection{Network Safety}

In a network, the dynamics of node $i$ can be modeled by
\begin{align}\label{eq:f(x,u,w)}
   \!\!\dot x_i(t) \!=\! f_i\big( x_i(t), u_i(t), w_i(t) \big),\, x_i \!\in\! \mathbb{R}^n, \, u_i \!\in\! U, \, w_i \!\in\! W,\!\!
\end{align}
with $u_i$ a control input, $w_i$ an external input, and $U$ and $W$ their respective admissible sets. For such a system we need to adapt our definition of invariant sets following \cite{invariance}. 

\begin{definition}\label{def:robust control invariant}
    A set $S$ is \emph{robust control invariant} by the dynamics \eqref{eq:f(x,u,w)} if there exists a feedback control law $u_i(t)$ such that for all $x_i(0) \in S$ and all time-varying $w_i \in W$, $x_i(t) \in S$ for all $t \geq 0$.
\end{definition}

We then want to determine the set of control inputs $u$ ensuring the robust control invariance of a given safe set $S$ despite the fluctuations in the neighboring inverter states. In particular, specific to the example of inverter-based microgrids, i.e., when the safe sets are expressed as box constraints on the states (as in \eqref{eq:safe set}), we define the following:

\begin{definition}\label{def:upper and lower invariant}
    A 1-dimensional set $S = [\underline{s}, \overline{s}]$ is \emph{upper invariant} (resp. \emph{lower invariant}) for a set of controls $U$ by the dynamics \eqref{eq:f(x,u,w)} if for all time-varying $w \in W$, $u \in U$ and all $x(0) \in S$, then $x(t) \leq \overline{s}$ $\big( \text{resp.}\ x(t) \geq \underline{s} \big)$ for $t \geq 0$.
\end{definition}

If $S$ is upper invariant (resp. lower invariant), then the state cannot escape by crossing the upper bound (resp. lower bound) of $S$.
Notice that if there are controls making $S$ both upper and lower invariant, then $S$ is a robust control invariant set. We denote such controls as \textit{safety admissible}.

\begin{definition}\label{def:safety admissible}
    A set of controls $U$ is \emph{safety admissible} for the set $S$ if for all controls $u \in U$, the set $S$ is robust control invariant.
\end{definition}

As we will detail later, if the dynamics \eqref{eq:f(x,u,w)} and the safe set $S$ are polynomial, e.g., $S = \big\{ x \in \mathbb{R}^n : p_j(x)~\geq~0, \break j~\in~\{1,\dots,m\} \big\}$ with $p_j$ polynomials, then the invariance condition of Nagumo's theorem can be stated as a polynomial inequality, enabling its fast computation.

\section{Problem Description}\label{sec:problem}

A microgrid power network is operated by a microgrid coordinator that determines power setpoints for each node of the grid \cite{kundu2020transient,almassalkhi2020hierarchical}. The transition in between these setpoints, corresponding to a transient regime, might lead the frequency or voltage of some inverters to violate safety constraints.
We are thus interested in ensuring the transient safety of microgrid networks, so that they are reliable when operated. Consider the case of droop-controlled inverters \cite{schiffer2014conditions}:
\begin{subequations}\label{eq:inverter}
\begin{align}
    \dot \theta_i &= \omega_i,   \label{eq:theta} \\
    \tau_i \dot \omega_i &= - \omega_i + \lambda_i^p \big( P_i^{set} - P_i \big),  \label{eq:frequency} \\
    \tau_i \dot v_i &= v_i^0 - v_i + \lambda_i^q \big( Q_i^{set} - Q_i \big), \label{eq:voltage}
\end{align}
\end{subequations}
where $\theta_i$, $\omega_i$ and $v_i$ are, respectively, the phase angle, frequency and voltage magnitude of node $i$. 
The droop-coefficients $\lambda_i^p > 0$ and $\lambda_i^q > 0$ are associated with the active power vs. frequency and the reactive power vs. voltage droop curves, respectively. The time-constant of the low-pass filter used for the active and reactive power measurements is $\tau_i$. The nominal voltage magnitude is $v_i^0$. The active power and reactive power set-points are $P_i^{set}$ and $Q_i^{set}$, respectively. Finally, the active and reactive power injected into the network are $P_i$ and $Q_i$, respectively following the nonlinear coupling equations
\begin{subequations}\label{eq:P and Q}
\begin{align}
    P_i &= v_i \sum_{k \in \mathcal{N}_i} v_k ( G_{i,k} \cos \theta_{k,i} - B_{i,k} \sin \theta_{k,i} ), \label{eq:P} \\
    Q_i &= -v_i \sum_{k \in \mathcal{N}_i} v_k ( G_{i,k} \sin \theta_{k,i} + B_{i,k} \cos \theta_{k,i} ), \label{eq:Q}
\end{align}
\end{subequations}
where $\theta_{k,i} = \theta_k - \theta_i$, and $\mathcal{N}_i$ is the set of neighbor nodes with the convention that $i \in \mathcal{N}_i$. The transfer conductance and susceptance of the line connecting nodes $i$ and $k$ are denoted by $G_{i,k}$ and $B_{i,k}$, respectively.

We use the formulation of \cite{kundu2019distributed} to model the capability of the inverters to change their control set-points of the active and reactive power output in order to adjust to different operating conditions. More specifically, we write
\begin{equation}\label{eq:set-points}
    P_i^{set} = P_i^0 + u_i^p, \quad \text{and} \quad  Q_i^{set} = Q_i^0 + u_i^q,
\end{equation}
where $P_i^0$ and $Q_i^0$ are the set-points for the nominal operating conditions; and $u_i^p$ and $u_i^q$ are control inputs.
%
We are thus interested in maintaining at all times both voltage and frequency within some pre-specified safety limits.
By a usual abuse of notation, instead of considering the actual voltage, we consider its difference from the nominal voltage $v_i^0 = 1$ p.u.. Following \cite{kundu2019distributed, schneider2018improving, elizondo2016inertial}, we consider the voltage and frequency safe sets to be 
\begin{subequations}
    \begin{align}
    S_v & = [\underline{v}, \overline{v}] = [-0.4, 0.2]\, \text{p.u.}\,, \\
    S_\omega & = [\underline{\omega}, \overline{\omega}] = [-3, 3]\, \text{Hz}\, .
\end{align}\label{eq:safe set}
\end{subequations}

For the inverters \eqref{eq:inverter}, the perturbation $w_i$ from Definition~\ref{def:robust control invariant} are the neighbor voltage magnitudes ($v_k$) and phase angle differences ($\theta_{k,i}$) that determine the power transfers, $P_i$ and $Q_i$\,. For the purpose of this paper, we assume that the phase angle differences between (any) two neighbor buses are bounded as follows:
\begin{align}
    \theta_{i,k} \in S_\theta = \left[-{\pi}/{6}, {\pi}/{6}\right]\quad\forall k\in\mathcal{N}_i\,.
\end{align}
Note that such range of phase angle differences are typical of distribution feeders, and especially microgrids, that are often characterized by relatively short lines connecting two buses \cite{kersting1991radial}. Furthermore, note that only the difference of phase angles $\theta_i$ and $\theta_k$ as opposed to their individual values, determine the power-flow connecting the nodes $i$ and $k$\,. Thus, for simplicity to notations and without any loss of generality, we set $\theta_i\equiv 0$ as the \textit{reference angle}, and use $\theta_{k,i}\equiv \theta_k$ throughout this text.
Then, we want to determine controls $u_i=\left(u_i^p,\,u_i^q\right)$ that maintain $v_i \in S_v$ and $\omega_i \in S_\omega$, whatever the values of $\theta_k \in S_\theta$ and $v_k \in S_v$ for the neighbors $k \in \mathcal{N}_i$.

\begin{problem}\label{prob:controls}
    \textit{(Safety-Admissible Control)} What values of control set-points $u_i$ ensure that $S_v$ and $S_\omega$ are robust control invariant by the dynamics \eqref{eq:inverter}?
\end{problem}

Moreover, note that the impact of neighbor (and network) disturbances on the inverter dynamics \eqref{eq:inverter} are enhanced by larger droop values, a fact which suggests the existence and size of the \textit{safety-admissible} controls depend on the chosen droop coefficients. This motivates the following question:

\begin{problem}\label{prob:droop}
    \textit{(Maximal Droop)} What values of droop coefficients $\left(\lambda_i^p,\lambda_i^q\right)$ ensure the existence of a non-empty safety-admissible set of controls, guaranteeing robust control invariance of $S_v$ and $S_\omega$ as per dynamics \eqref{eq:inverter}?
\end{problem}

\section{Theoretical Construction}\label{sec:theory}

In this section we establish the theory concerning robust control invariant sets for droop-controlled inverters.
%
%
%
%
Since $\lambda_i^p > 0$ and $\lambda_i^q > 0$ in the inverter dynamics \eqref{eq:inverter}, we can define a \emph{minimal lower control} $\underline{u}$ and a \emph{maximal upper control} $\overline{u}$ such that
\begin{align*}
    \underline{u} &= \inf\big\{ u_{low} \in \mathbb{R} : S\ \text{is lower invariant for all}\ u \geq u_{low} \big\}, \\
    \overline{u} &= \sup\big\{ u_{up} \in \mathbb{R} : S\ \text{is upper invariant for all}\ u \leq u_{up} \big\}.
\end{align*}
If $\underline{u} \leq \overline{u}$, then the maximal interval of safety admissible controls is $[\underline{u}, \overline{u}]$.
To illustrate these definitions and our objective let us study a simplified version of \eqref{eq:inverter}.

\subsection{Simple Example}\label{subsec:example}

Consider the frequency dynamics equation \eqref{eq:frequency} when the inverter has a single neighbor. To simplify take $\tau = 1$s, $\lambda^p = 1$rad/s/p.u., $\theta_2 = 0$rad, $P^0 = 0$p.u., $G = -2$p.u.$^{-1}$ and a safe set $S = [-1, 1]$Hz. Then, $\dot \omega = - \omega + 2\omega \omega_2 + u$.

To make $S$ upper invariant, according to Nagumo's theorem we need $\dot \omega \leq 0$ when $\omega = 1$. Thus, we want $2\omega_2-1 + u \leq 0$, so $u \leq 1 -2\omega_2$. We are looking for robust controls $u$ working for all possible $\omega_2 \in S$. Then, the maximal upper control is $\overline{u} = -1$, because for $\omega_2 = 1$, we need $u \leq -1$.

Similarly, to make $S$ lower invariant, we need $\dot \omega \geq 0$ when $\omega = -1$. Then, we want $-2\omega_2 +1 + u \geq 0$, i.e., $u \geq -1 + 2\omega_2$. Thus, the minimal lower control is $\underline{u} = 1$.

In this setting, $\overline{u} < \underline{u}$, there are no safety admissible controls making $S$ robust control invariant. A reason for this failure is that $\lambda^p$ is too large, making $\omega$ unstable. Indeed, for small values of $\lambda^p$, \eqref{eq:frequency} can be approximated by $\tau_i \dot \omega_i = -\omega_i$ which is stable. We will elaborate further on this issue in the following subsections.

\subsection{Minimal controls for upper and lower invariance}

Let us determine the minimal lower control for $S_\omega$.
By definition, $\underline{u_i^p} = \underset{\theta_k, v_k}{\inf} \{ u_i^p : \dot \omega_i \geq 0,\ \omega_i = \underline{\omega} \}$. With \eqref{eq:frequency}, $\underline{u_i^p} = \inf \big\{ u_i^p : u_i^p \geq \frac{1}{\lambda_i^p} \underline{\omega} + P_i - P^0_i,\ \forall \theta_k, v_k \big\}$. Then, the minimal $u_i^p$ that is larger than all $\big[\frac{1}{\lambda_i^p} \underline{\omega} + P_i - P^0_i\big] (\theta_k, v_k)$ is in fact the maximum of this term over all $\theta_k$ and $v_k$, because $P_i$ is continuous in $\theta_k, v_k$ according to \eqref{eq:P} and, $S_\theta$ and $S_v$ are compact. Thus,
\begin{align}\label{eq:min lower}
    &\underline{u^p_i} =  \underset{\theta_k, v_k}{\max} \ \frac{1}{\lambda_i^p}\underline{\omega} + P_i - P_i^0,\\
    & \text{subject to}\quad \theta_k \in S_\theta,\ v_k \in S_v,\ k \in \mathcal{N}_i, \nonumber
\end{align}
Then, $\dot \omega_i \geq 0$ when $\omega_i = \underline{\omega}$ and $u_i^p \geq \underline{u^p_i}$ for all $\theta \in S_\theta$ and $v \in S_v$, which guarantees lower invariance according to Nagumo's theorem. Similarly, the maximal upper control is defined as
\begin{align}\label{eq:max upper}
   & \overline{u^p_i} =  \underset{\theta_k, v_k}{\min} \ \frac{1}{\lambda_i^p}\overline{\omega} + P_i - P_i^0,\\
    & \text{subject to}\quad \theta_k \in S_\theta,\ v_k \in S_v,\ k \in \mathcal{N}_i, \nonumber
\end{align}
and makes $\dot \omega_i \leq 0$ when $\omega_i = \overline{\omega}$ and $u_i^p \leq \overline{u^p_i}$ for all $\theta \in S_\theta$ and $v \in S_v$. The minimal lower control $\underline{u_i^q}$ and maximal upper control $\overline{u_i^q}$ for the voltage equation are defined similarly.

A great way to solve the optimization problems \eqref{eq:min lower} and \eqref{eq:max upper} is to use a sum-of-squares (SOS) optimization.
A multivariate polynomial $p(x)$, $x \in \mathbb{R}^n$, is an SOS if there exist polynomial functions $h_i(x)$, $i = 1, \hdots, s$ such that $p(x) = \sum_{i=1}^s h_i^2(x)$.
However, the power-flow equations \eqref{eq:P and Q} are not polynomials. Following \cite{kundu2019distributed} we choose a third order Taylor expansion of the dynamics \eqref{eq:P and Q} to make \eqref{eq:inverter} polynomial.
Then, \eqref{eq:min lower} and \eqref{eq:max upper} can be solved with any SOS tool.

\subsection{Maximal droop for robust control invariance}

The stability of the voltage and frequency rely on small droop coefficients $\lambda$. As in Section~\ref{subsec:example}, when $\lambda$ increases in \eqref{eq:inverter}, the perturbations $P$ and $Q$ increase, and we thus have the intuition that the set of safety admissible controls should shrink. We can actually prove a stronger result by using the fact that the optimizations in \eqref{eq:min lower} and \eqref{eq:max upper} only affect $P_i$.

\vspace{1mm}

\begin{proposition}\label{prop:monotonic range}
    The bounds of the interval of safety admissible controls $[\underline{u}, \overline{u}]$ for $S_v$ and $S_\omega$ are inversely proportional to the droop coefficient $\lambda$.
\end{proposition}
\begin{proof}
    We first introduce the maximum and minimum of the active and reactive powers \eqref{eq:P and Q} at the lower and upper bounds of $S_\omega$ and $S_v$. More specifically
    \begin{align*}
        P_i^{max} &= \max\big\{ P_i : \omega_i = \underline{\omega},\ (\theta_k, v_k) \in S_\theta \times S_v,\ k \in \mathcal{N}_i \big\}, \\
        P_i^{min} &= \min\big\{ P_i : \omega_i = \overline{\omega},\ (\theta_k, v_k) \in S_\theta \times S_v,\ k \in \mathcal{N}_i \big\}, \\
        Q_i^{max} &= \max\big\{ Q_i : v_i = \underline{v},\ (\theta_k, v_k) \in S_\theta \times S_v,\ k \in \mathcal{N}_i \big\}, \\
        Q_i^{min} &= \min\big\{ Q_i : v_i = \overline{v},\ (\theta_k, v_k) \in S_\theta \times S_v,\ k \in \mathcal{N}_i \big\}.
    \end{align*}
    Then, \eqref{eq:min lower} and \eqref{eq:max upper} simplify to
    \begin{equation}\label{eq:safety admissible with P^max}
        \underline{u_i^p}(\lambda_i^p) = \frac{1}{\lambda_i^p}\underline{\omega} + P_i^{max} \hspace{-1mm} - P_i^0,\ \overline{u_i^p}(\lambda_i^p) = \frac{1}{\lambda_i^p}\overline{\omega} + P_i^{min} \hspace{-1mm} - P_i^0.
    \end{equation}
    A similar result holds for the voltage.
\end{proof}

\vspace{1mm}

Building on this result, we can then establish a sufficient condition for the existence of safety admissible controls.

\vspace{1mm}

\begin{proposition}\label{prop:0 in interior}
    If the safe set contains $0$ in its interior, then safety admissible controls exist for some droop coefficients.
\end{proposition}
\begin{proof}
    Note that $\omega_i$ does not intervene in $P_i$ \eqref{eq:P}, so the constraints for $P_i^{min}$ and $P_i^{max}$ are the same, which leads to $P_i^{min} < P_i^{max}$. 
    Then, based on \eqref{eq:safety admissible with P^max}, the condition $\underline{\omega} < 0 < \overline{\omega}$ is necessary and sufficient to make $\underline{u_i^p}(\lambda_i^p) < \overline{u_i^p}(\lambda_i^p)$ for $\lambda_i^p$ small enough. 

    On the other hand, since $v_i$ intervenes in $Q_i$ \eqref{eq:Q}, we cannot order $Q_i^{min}$ and $Q_i^{max}$ without computing them. Thus the condition $\underline{v} < 0 < \overline{v}$ is sufficient but maybe not necessary to make $\underline{u_i^q}(\lambda_i^q) < \overline{u_i^q}(\lambda_i^q)$ for some $\lambda_i^q$.
\end{proof}

\vspace{1mm}

Then, the safe sets of \eqref{eq:safe set} guarantee that safety admissible controls exist for some small enough droop coefficients.
We now want to find the maximal droop coefficient $\lambda^*$ for which safety admissible controls exist, i.e.,
\begin{equation}\label{eq:lambda*}
    \lambda^* = \max \big\{ \lambda \geq 0 : \underline{u} (\lambda) \leq \overline{u} (\lambda) \big\}.
\end{equation}
\begin{remark}
Note that the problem of identifying the maximal droop for stability analysis is relatively well studied in the literature (see, for instance, \cite{schiffer2014conditions,kundu2019identifying,nandanoori2020distributed}). However, as a novel contribution of this paper, we propose a method to identify maximal droop values for safety verification.
\end{remark} 
According to Proposition~\ref{prop:monotonic range} and \ref{prop:0 in interior}, if $0$ is in the interior of $S$ and $\lambda \leq \lambda^*$, then the interval $[\underline{u}(\lambda), \overline{u}(\lambda)]$ is not empty, is proportional with $1/\lambda$ and makes $S$ robust control invariant,.

\vspace{1mm}

\begin{proposition}\label{prop:lambda*}
    The maximal droop coefficient $\lambda^*$ is
    \begin{equation}\label{eq:lambda* with P and Q}
        \lambda_i^{p*} = \frac{\overline{\omega} - \underline{\omega}}{P_i^{max} - P_i^{min}} \quad \text{and} \quad \lambda_i^{q*} = \frac{\overline{v} - \underline{v}}{Q_i^{max} - Q_i^{min}}.
    \end{equation}
\end{proposition}
\begin{proof}
    Since $\underline{u_i^p}$ and $\overline{u_i^p}$ are continuous in $\lambda_i^p$ according to \eqref{eq:safety admissible with P^max}, definition \eqref{eq:lambda*} leads to $\underline{u_i^p}(\lambda_i^{p*}) = \overline{u_i^p}(\lambda_i^{p*})$. Using \eqref{eq:safety admissible with P^max} we easily obtain the announced expression for $\lambda_i^{p*}$. The calculation of $\lambda_i^{q*}$ is similar.
\end{proof}

\vspace{1mm}


The term $P^{max} - P^{min}$ in the denominator validates our intuition that increasing the range of possible perturbations decreases $\lambda^{p*}$. The larger $\overline{\omega} - \underline{\omega}$ is, the larger $\omega$ can be, and thus the stronger the stabilizing term $-\omega$ is in \eqref{eq:frequency}, which increases $\lambda^{p*}$. The same analysis holds for $\lambda^{q*}$.

\begin{remark}
There are now two approaches to answer Problem~\ref{prob:controls} based on the controls $\underline{u}$ and $\overline{u}$. If $\lambda \leq \lambda^*$, then $\underline{u} \leq \overline{u}$ and any control in between guarantees the robust control invariance of the safe set. On the other hand, if $\lambda > \lambda^*$ we need a state-dependent control law. When the state gets too close from the upper (resp. lower) bound of its safe set, applying $\overline{u}$ (resp. $\underline{u}$) prevents safety violation. Besides, $\underline{u}$ and $\overline{u}$ can be precomputed, so that the sole real-time action of the controller is to decide which one of the controls to apply. 
\end{remark}

\section{Numerical Analysis}\label{sec:example}

In this section we apply our theory to a modified microgrid model \cite{ersal2011impact}, and demonstrate the robust control invariance of the safe sets $S_v$ and $S_\omega$ introduced in \eqref{eq:safe set}. As per the modifications in \cite{kundu2019distributed,kundu2020transient}, four inverters were placed in the network, three of those at buses 1,\,4, and 5, and the fourth at bus 0 after disconnecting the utility for islanded operation.

In a microgrid, the distance between inverters is small and thus the states of neighbors are strongly coupled. To account for this physical phenomenon we introduce two constants $\Delta_v$ and $\Delta_\omega$ measuring the range of allowable uncertainty of neighboring inverters such that $v_k \in [v_i - \Delta_v, v_i + \Delta_v]$ and $\omega_k \in [\omega_i - \Delta_\omega, \omega_i + \Delta_\omega]$ for $k \in \mathcal{N}_i$. For the numerical analysis we choose $\Delta_v = 0.02$p.u., which is $2\%$ of the nominal voltage $v_i^0$, and $\Delta_\omega = 0.12$Hz, which is $2\%$ of the $6$Hz range of $S_\omega$. The nominal droop coefficients of the network are $\lambda^p = 2.51$rad/s/p.u. and $\lambda^q = 0.2$p.u./p.u..

The coupling adds a constraint to the calculation of $Q_i^{min}$ and $Q_i^{max}$. For instance, $Q_i^{max} = \max\ Q_i(\theta_k, v_k)$ subject to $v_i = \underline{v}$, $\theta_k \in S_\theta$, $v_k \in S_v \cap [v_i - \Delta_v, v_i + \Delta_v] = [\underline{v}, \underline{v} + \Delta_v]$ for $k \in \mathcal{N}_i$. We compute $Q_2^{min}$, $Q_2^{max}$ with an SOS algorithm and we use the calculations of Proposition~\ref{prop:monotonic range} to represent $\underline{u_2^q}$ and $\overline{u_2^q}$ on Figure~\ref{fig:range_v2}.

\begin{figure}[htbp!]
    \centering
    \includegraphics[scale=0.4]{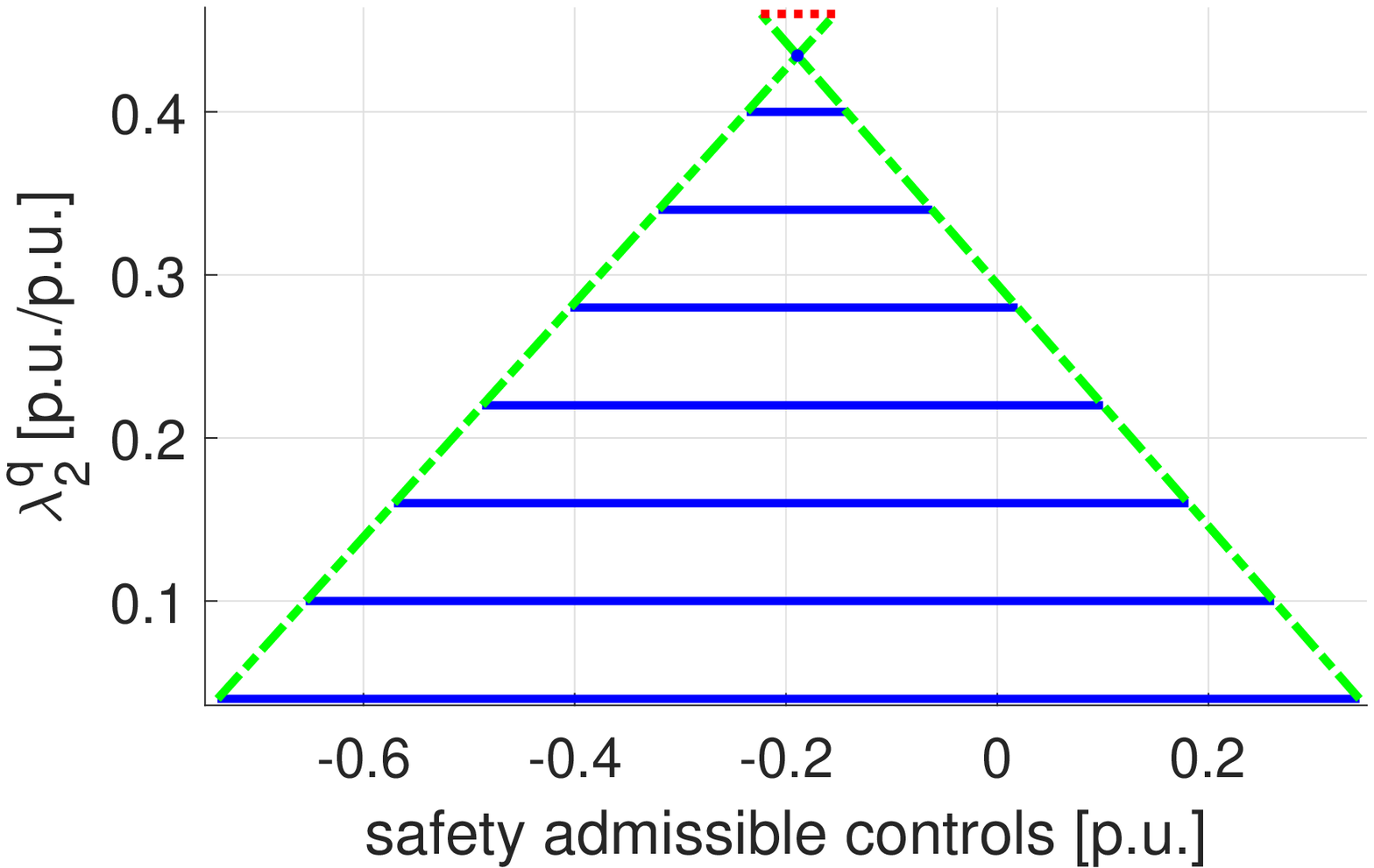}
    \caption{Evolution of the interval of safety admissible controls in blue for the voltage of node 2. The blue dot corresponds to $\lambda_2^{q*}$. The green dash-dot lines correspond to $\underline{u_2^q}(\lambda_2^q)$ and $\overline{u_2^q}(\lambda_2^q)$. The red dotted line indicates that $\underline{u_2^q} > \overline{u_2^q}$.}
    \label{fig:range_v2}
\end{figure}

\begin{figure}[htbp!]
    \centering
    \includegraphics[scale=0.4]{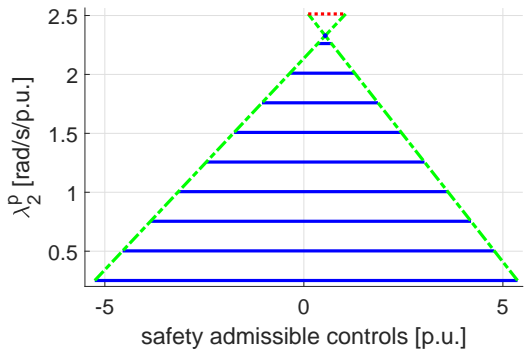}
    \caption{Evolution of the interval of safety admissible controls in blue for the frequency of node 2. The blue dot corresponds to $\lambda_2^{p*}$. The green dash-dot lines correspond to $\underline{u_2^p}(\lambda_2^p)$ and $\overline{u_2^p}(\lambda_2^p)$. The red dotted line indicates that $\underline{u_2^p} > \overline{u_2^p}$.}
    \label{fig:range_omega2}
\end{figure}

As proven in Proposition~\ref{prop:0 in interior}, since $0$ is in the interior of $S_v$, we were able to find values of $\lambda_2^q$ for which safety admissible controls exist. On Figure~\ref{fig:range_v2}, $\lambda_2^{q*}$ is located at the intersection of the green dash-dot lines representing $\underline{u_2^q}(\lambda_2^q)$ and $\overline{u_2^q}(\lambda_2^q)$.
Figure~\ref{fig:range_omega2} shows that the same is true for $\lambda_2^p$.
The computed values of $\lambda^*$ are gathered in Table~\ref{tab:lambda*}.

\renewcommand{\arraystretch}{1.2}
\begin{table}[htbp!]
    \centering
    \begin{tabular}{c|cccc}
        Inverter & 1 & 2 & 3 & 4 \\ \hline
        $\lambda^{p*}$ [rad/s/p.u.] & $1.227$ & $2.329$ & $0.875$ & $1.368$ \\
        $\lambda^{q*}$ [p.u./p.u.] & $0.228$ & $0.434$ & $0.161$ & $0.241$
    \end{tabular}
    \caption{Maximal droop coefficients $\lambda^*$ for which safety admissible controls exist.}
    \label{tab:lambda*}
\end{table}

Note that for the voltage $\lambda^{q*}$ ranges from $80\%$ to $217\%$ of the nominal $\lambda^q$ depending on the inverter. For the frequency, $\lambda^{p*}$ ranges from $35\%$ to $93\%$ of the nominal $\lambda^p$.

We now study how the admissible range of states $\Delta_v$, $S_\theta$ and $S_v$ impact the maximal droop coefficient $\lambda^{q*}$.

\begin{figure}[htbp!]
    \centering
    \includegraphics[scale=0.4]{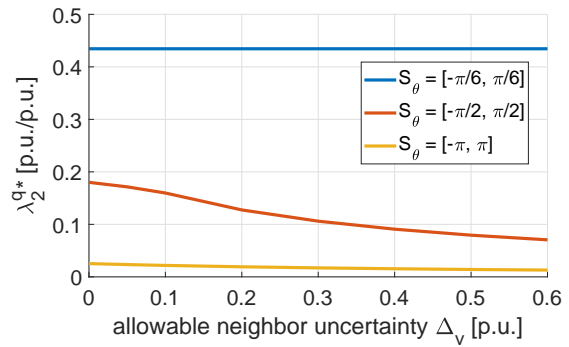}
    \caption{Evolution of $\lambda_2^{q*}$ as a function of the allowable neighbor uncertainty $\Delta_v$ and of the range of admissible phase angle $S_\theta$.}
    \label{fig:lambda_star_theta}
\end{figure}

In Figures~\ref{fig:lambda_star_theta} and \ref{fig:lambda_star_v}, $\Delta_v = 0.6$p.u. depicts a lack of coupling between inverters, because the length of $S_v$ is at most $0.6$p.u., while for $\Delta_v = 0$p.u. the coupling is perfect, i.e., $v_k = v_i$ for $k \in \mathcal{N}_i$.

\begin{figure}[htbp!]
    \centering
    \includegraphics[scale=0.42]{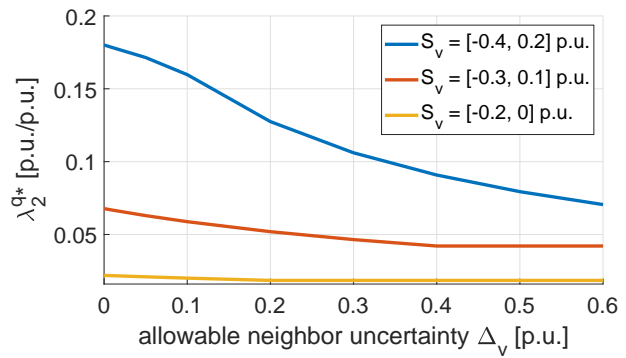}
    \caption{Evolution of $\lambda_2^{q*}$ as a function of the allowable neighbor uncertainty $\Delta_v$ and of the range of admissible voltage $S_v$.}
    \label{fig:lambda_star_v}
\end{figure}

As predicted with \eqref{eq:lambda* with P and Q} and illustrated on Figures~\ref{fig:lambda_star_theta} and \ref{fig:lambda_star_v}, as the uncertainty ranges $\Delta_v$ and $S_\theta$ increase, the value of $\lambda^{q*}$ decreases.
The impact of the length of $S_v$ on $\lambda^{q*}$ is more difficult to predict as it affects both the numerator and denominator of $\lambda^{q*}$ in \eqref{eq:lambda* with P and Q}. However, the verdict of Figure~\ref{fig:lambda_star_v} is clear: enlarging $S_v$ increases $\lambda^{q*}$, the voltage has more wiggle room inside its safe set, it is thus easier to maintain $v \in S_v$. 

\vspace{2mm}

We compute the controls $\underline{u^p_1}$ and $\overline{u_1^p}$ on an Intel Core i7-4770S with a CPU at 3.1GHz and 8GB of RAM.
Previous works \cite{kundu2019distributed, kundu2019identifying} relied on the widely used MATLAB toolbox SOSTOOLS \cite{sostools} and the semi-definite programming (SDP) solver SeDuMi \cite{Sedumi}. However, the computation times were often excessive for nodes with more than two neighbors. We thus consider the Julia language \cite{Julia}, its SOS toolbox \cite{sosJulia} and the SDP solvers SDPA \cite{SDPA} and Mosek \cite{Mosek}.
As expected, Julia is faster and the computation times can be reduced by two orders of magnitude as shown on Table~\ref{tab:speed}.

\renewcommand{\arraystretch}{1.5}
\begin{table}[htbp!]
    \centering
    \begin{tabular}{c|ccc}
        Language & MATLAB & Julia & Julia \\
        SDP solver & SeDuMi & SDPA & Mosek \\ \hline
        Run-time for $\underline{u^p_1}$ and $\overline{u_1^p}$ &  4295s & 343s & 33s \\ 
    \end{tabular}
    \caption{Run-time comparison between implementations.}
    \label{tab:speed}
\end{table}

In order to illustrate the invariance of the safe sets with the safety admissible controls calculated, we simulate the evolution of the voltage and the frequency of node 1. We use the original non-polynomial dynamics of the system. 
We choose $\lambda_1^p$ and $\lambda_1^q$ at respectively $40\%$ and $100\%$ of their nominal values, so that $\lambda_1^p < \lambda_1^{p*}$ and $\lambda_1^q < \lambda_1^{q*}$. Then, the intervals of safety admissible controls are $U_\omega = [-0.724, 1.571]$p.u. and $U_v = [-0.25, -0.0937]$p.u..

Every second we randomly choose a control value $u_1^q \in U_v$ and $u_1^p \in U_\omega$ as shown on Figures~\ref{fig:simulation}(\subref{fig:controls_voltage}) and \ref{fig:simulation}(\subref{fig:controls_frequency}).
Similarly, the states $\omega_k \in S_\omega \cap [\omega_1 - \Delta_\omega, \omega_1 + \Delta_\omega]$ and $v_k \in S_v \cap [v_1 - \Delta_v, v_1 + \Delta_v]$ of the neighboring nodes are stochastically updated every $10$ms as depicted on Figures~\ref{fig:simulation}(\subref{fig:neighbors_omega}) and \ref{fig:simulation}(\subref{fig:neighbors_v}). 
The evolution of the voltage and frequency of node 1 are then pictured on Figures~\ref{fig:simulation}(\subref{fig:v1}) and \ref{fig:simulation}(\subref{fig:w1}).

\begin{figure*}[htbp!]
\centering
\begin{subfigure}[b]{0.33\textwidth}
    \includegraphics[scale=0.4]{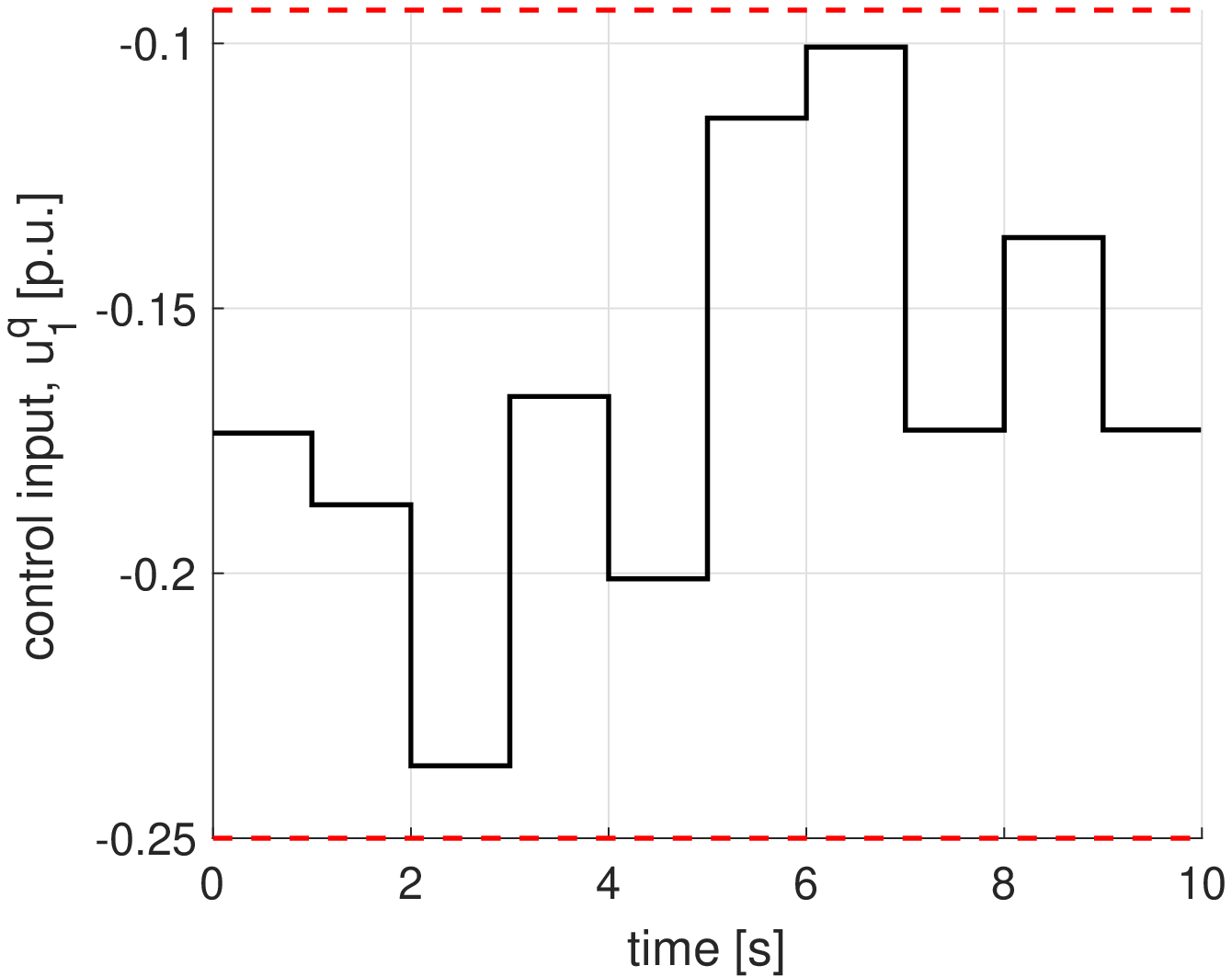}
    \caption{Stochastic choice of $u_1^q \in U_v$, where $U_v$ is the interval of safety admissible controls.}
    \label{fig:controls_voltage}
\end{subfigure}
\begin{subfigure}[b]{0.33\textwidth}
    \includegraphics[scale=0.4]{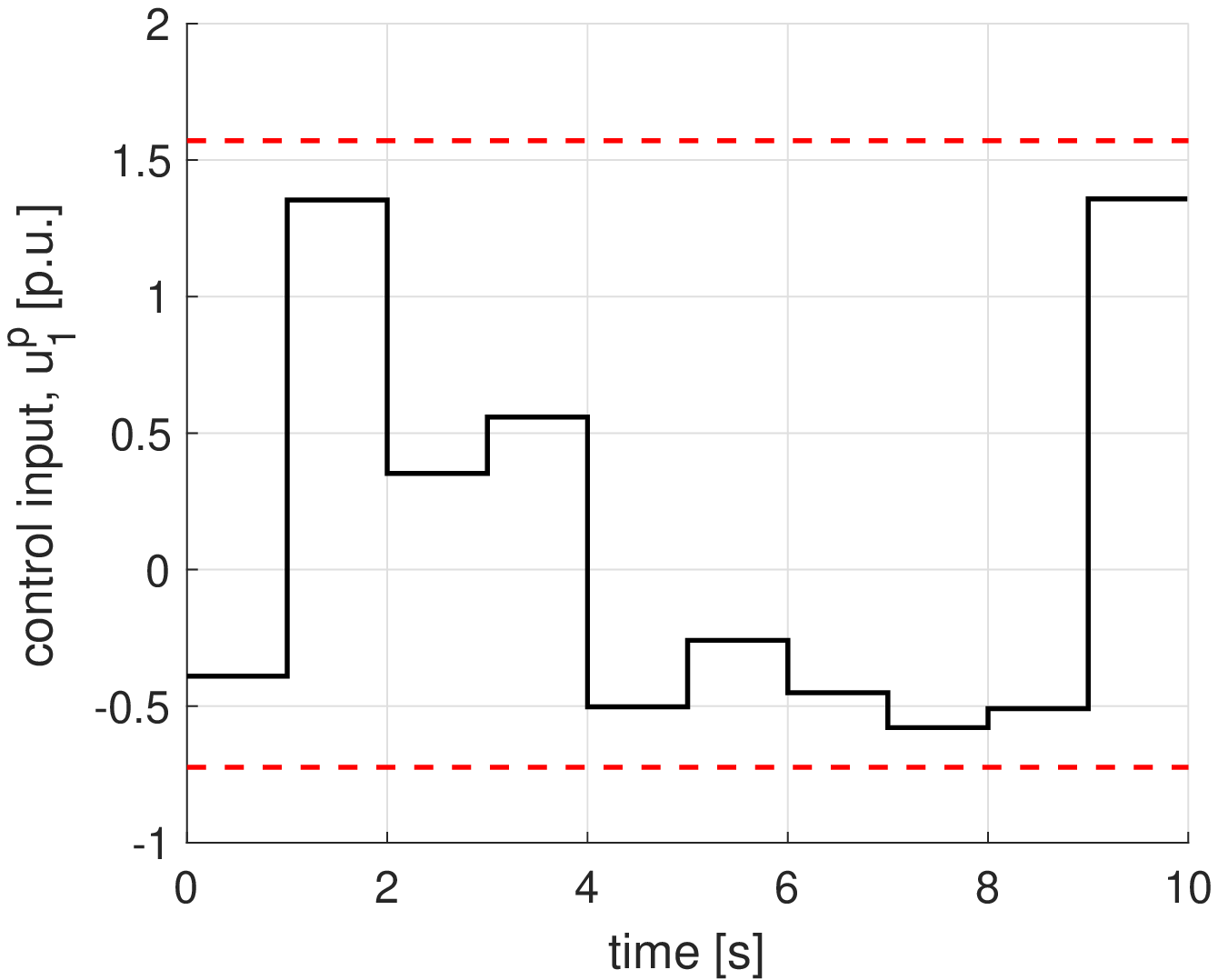}
    \caption{Stochastic choice of $u_1^p \in U_\omega$, where $U_\omega$ is the interval of safety admissible controls.}
    \label{fig:controls_frequency}
\end{subfigure}
\begin{subfigure}[b]{0.32\textwidth}
    \includegraphics[scale=0.4]{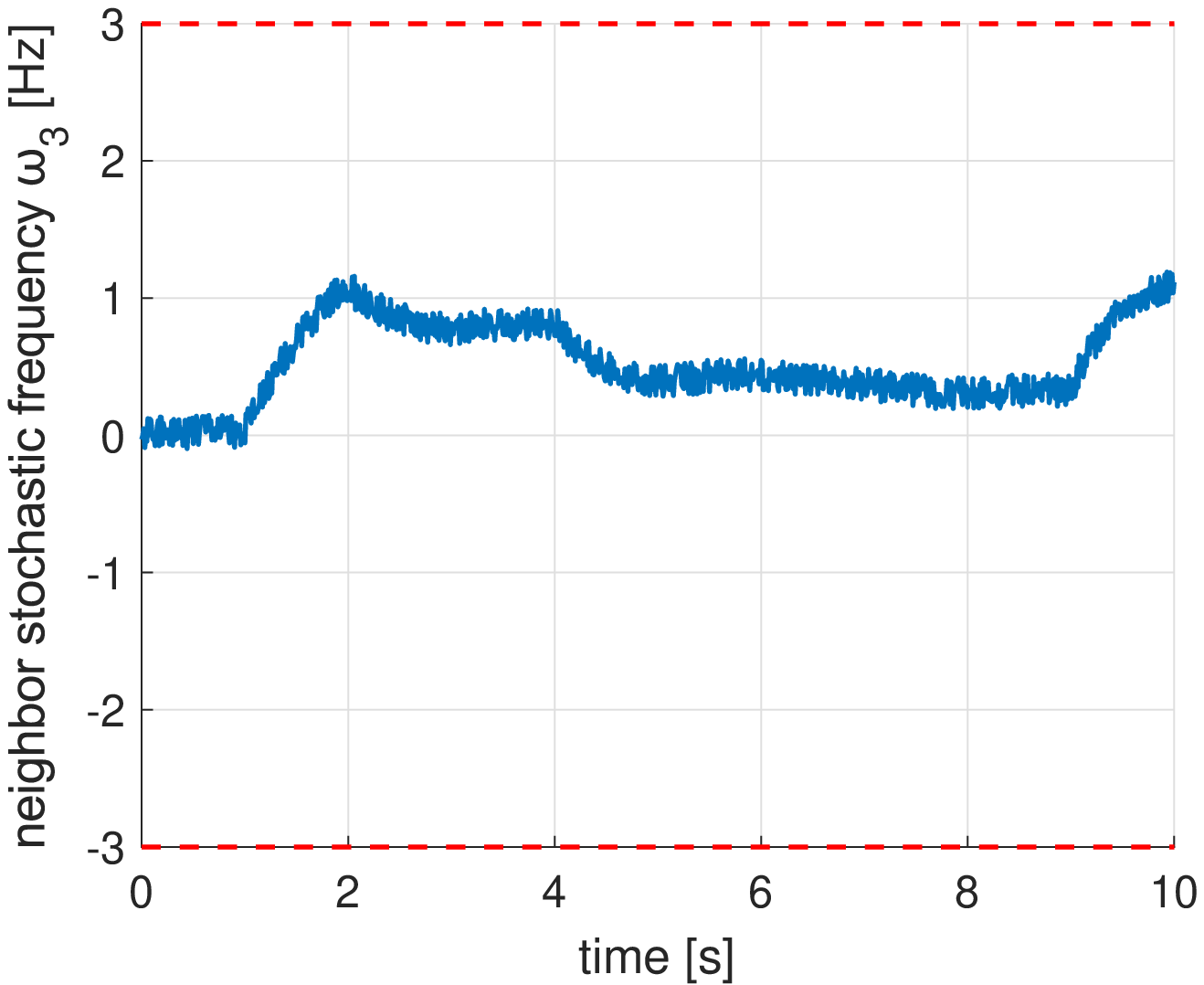}
    \caption{Stochastic choice of neighbor frequency $\omega_3 \in S_\omega \cap [\omega_1 - \Delta_\omega, \omega_1 + \Delta_\omega]$.}
    \label{fig:neighbors_omega}
\end{subfigure}

\medskip

\centering
\begin{subfigure}[b]{0.33\textwidth}
    \includegraphics[scale=0.4]{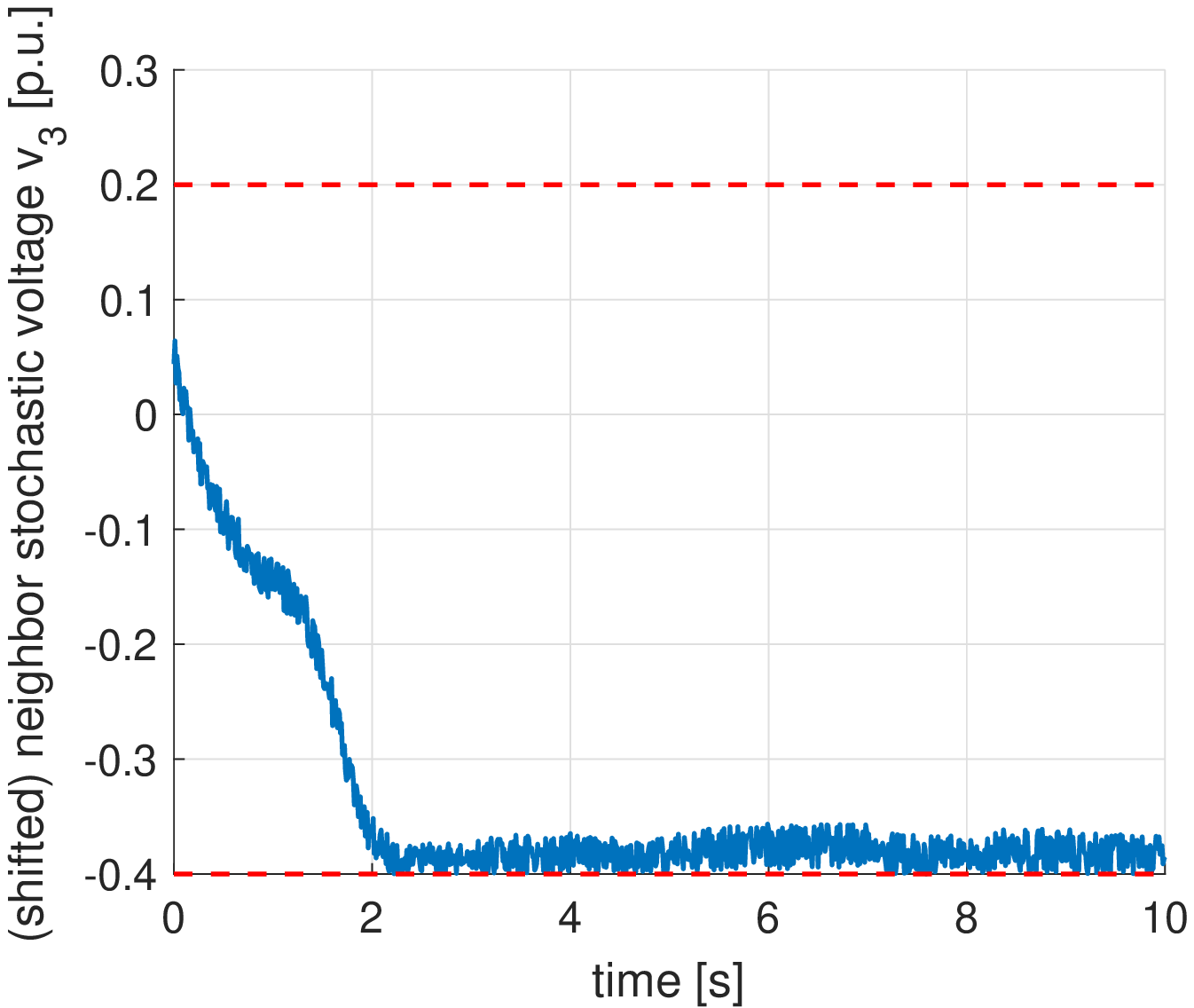}
    \caption{Stochastic choice of neighbor voltage\break
    $v_3 \in S_v \cap [v_1 - \Delta_v, v_1 + \Delta_v]$.}
    \label{fig:neighbors_v}
\end{subfigure}
\begin{subfigure}[b]{0.33\textwidth}
    \includegraphics[scale=0.4]{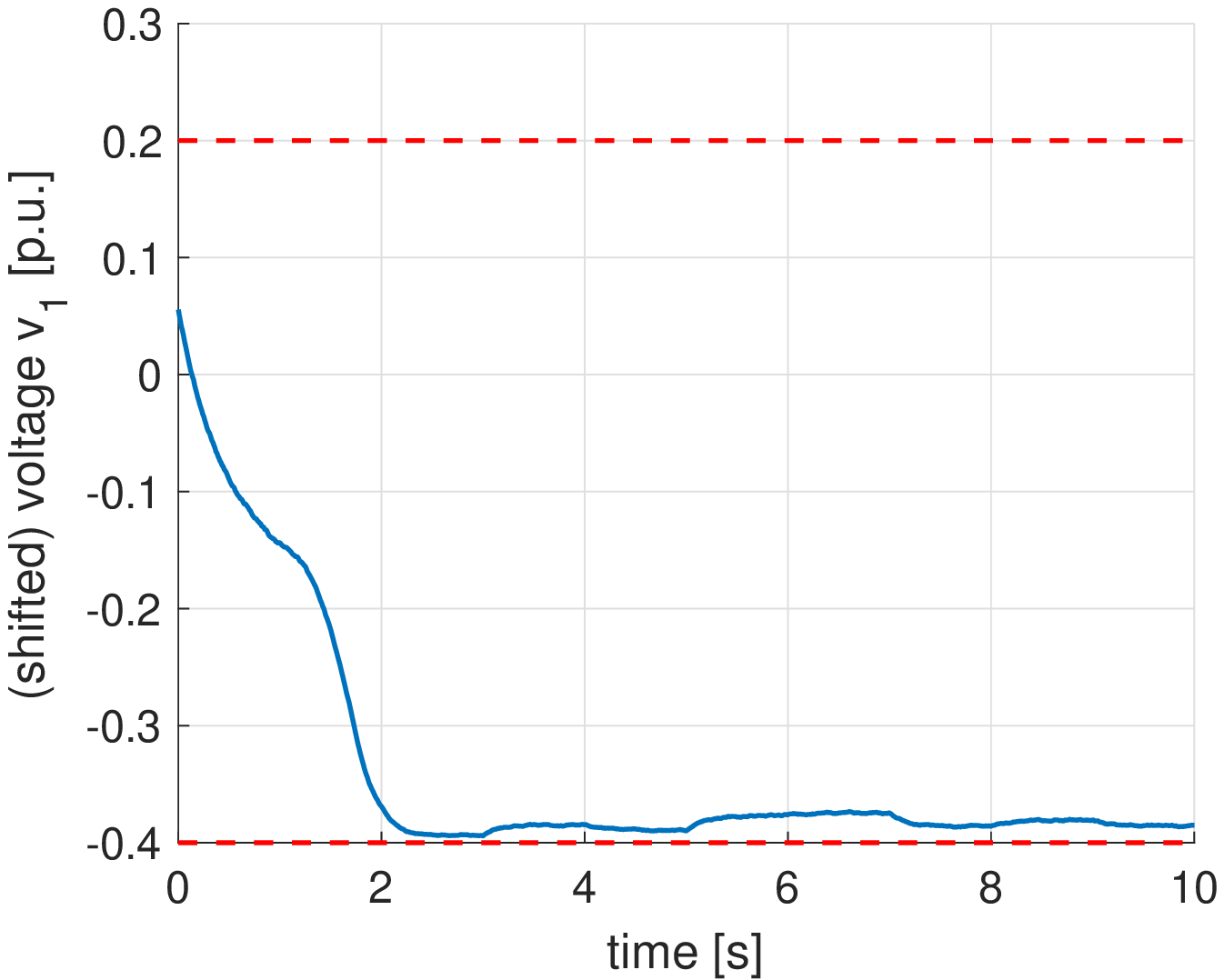}
    \caption{Evolution of $v_1$, kept in $S_v$ by $u_1^q$.}
    \label{fig:v1}
\end{subfigure}
\begin{subfigure}[b]{0.32\textwidth}
    \includegraphics[scale=0.4]{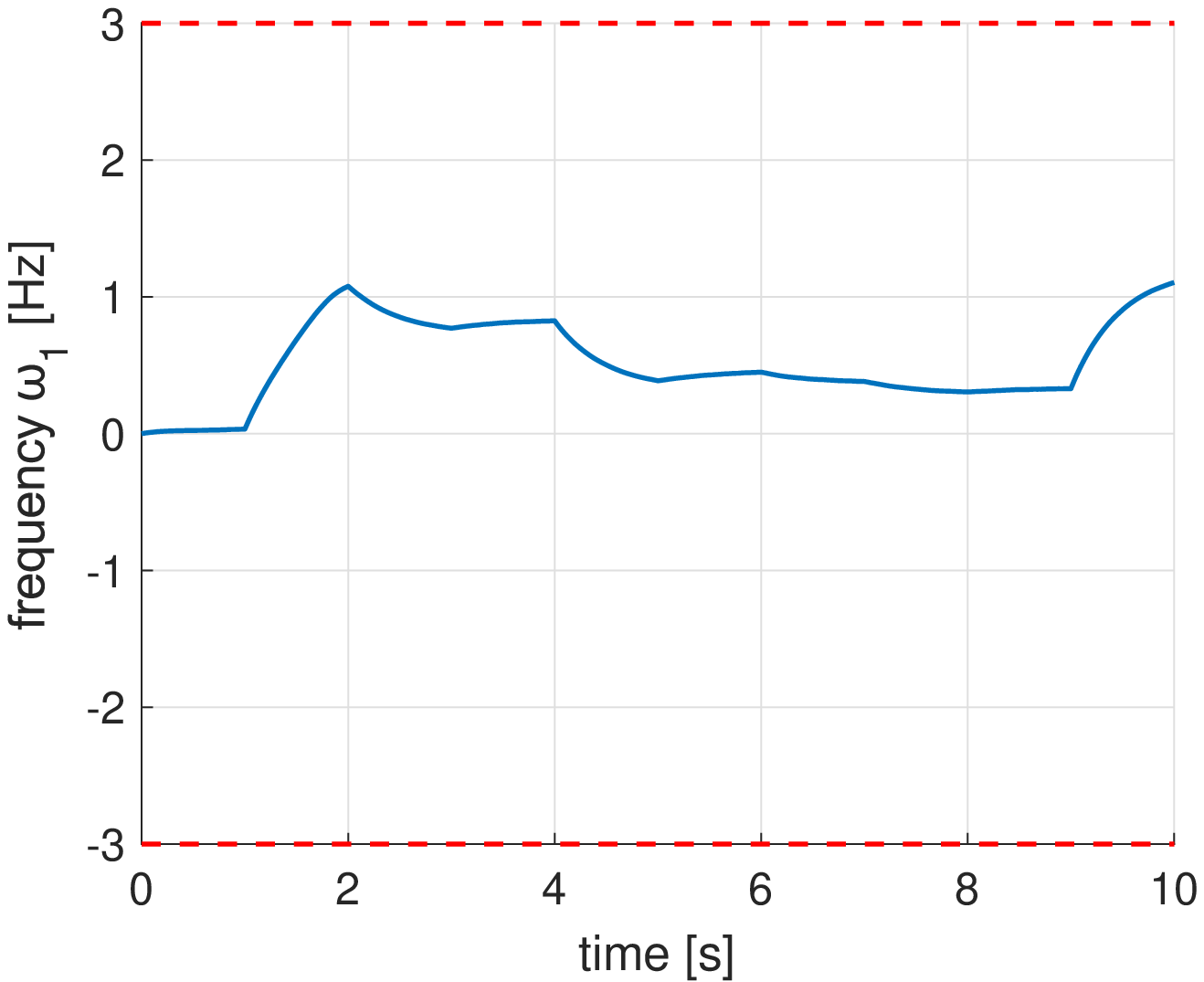}
    \caption{Evolution of $\omega_1$, kept in $S_\omega$ by $u_1^p$.}
    \label{fig:w1}
\end{subfigure}
\caption{Simulation of the voltage $v_1$ and frequency $\omega_1$ under stochastic safety admissible controls $u_1^q$ and $u_1^p$, and stochastic variations of neighbor states $\theta_2$, $\theta_3$, $v_2$ and $v_3$.}
\label{fig:simulation}
\end{figure*}

We can see that even randomly chosen controls, as long as they are within the safety admissible interval, enforce the safety of the system as $v_1 \in S_v$ and $\omega_1 \in S_\omega$ despite the stochastic variations of the neighbor states.

\vspace{3mm}

One could rightfully object that the stochastic nature of the variations of the neighbor states $\omega_k$ and $v_k$ prevent significant changes in $v_1$ and $\omega_1$ that could lead to safety violations. To overcome this limitation, we run a similar simulation where the neighbor states take their worst case values. To keep a constant phase angle $\theta_k = -\frac{\pi}{6}$rad, we need a constant frequency $\omega_k = 0$Hz and thus the frequency coupling must be removed by taking $\Delta_\omega = \infty$. 
We set the voltage at its lowest admissible bound, i.e., $v_k = \max\{ \underline{v}, v_i - \Delta_v\}$. We keep the same stochastic controls and same values of droop coefficients. 

Figure~\ref{fig:simulation_lb}(\subref{fig:v3_lb}) shows how the voltage of a neighboring inverter $v_3$ follows its admissible lower bound. As illustrated on Figure~\ref{fig:simulation_lb}(\subref{fig:v1_lb}) and \ref{fig:simulation_lb}(\subref{fig:w1_lb}) the voltage $v_1$ is maintained in $S_v$ and the frequency $\omega_1$ is maintained in $S_\omega$ despite the stochastic safety admissible controls and the lower bound neighbor states. Similar results are obtained when choosing $\theta_k$ at its upper bound $\frac{\pi}{6}$ and/or $v_k$ at its upper bound $\min\{ \overline{v}, v_i + \Delta_v \}$.

\begin{figure*}[htbp!]
\centering
\begin{subfigure}[b]{0.33\textwidth}
    \includegraphics[scale=0.4]{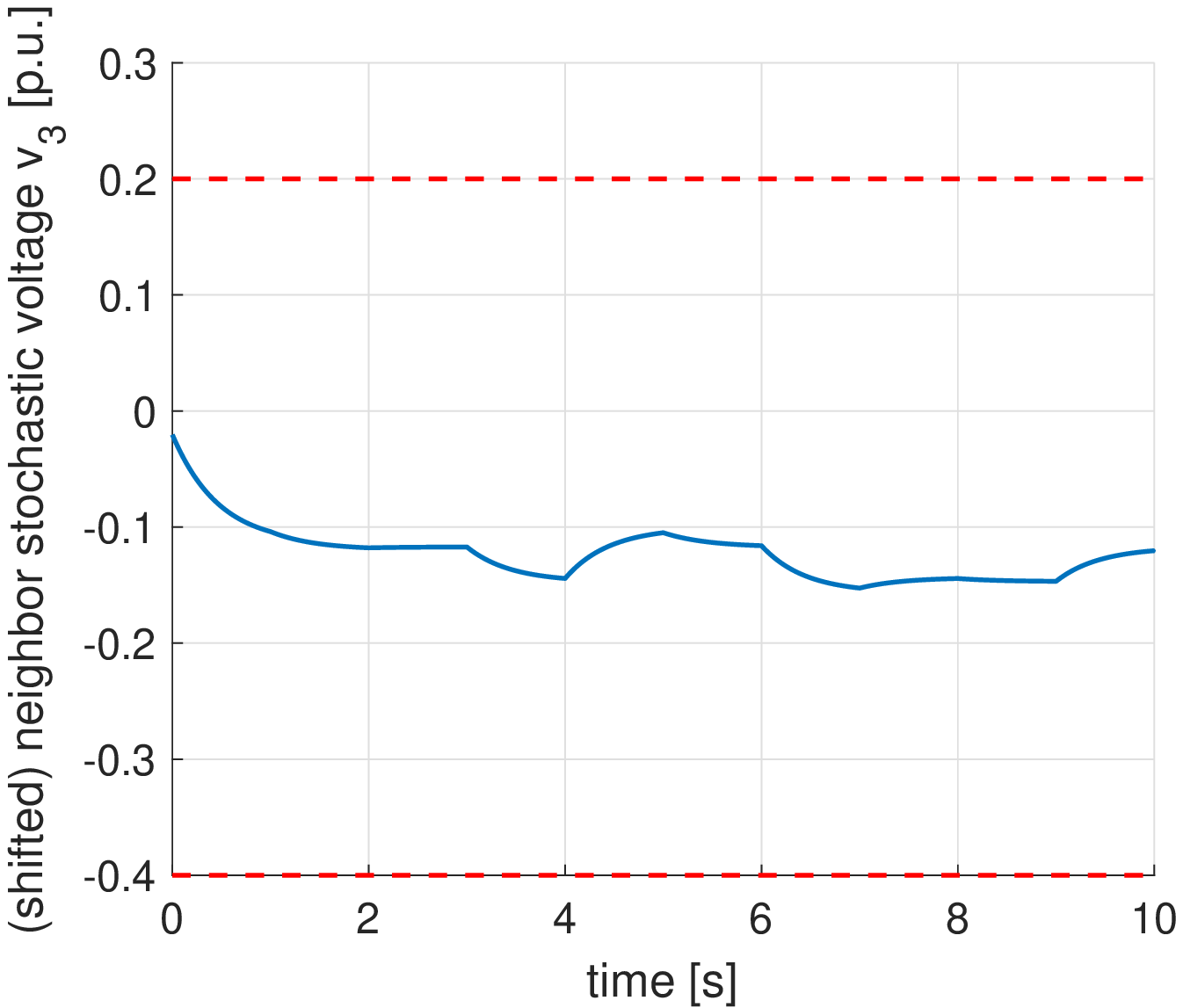}
    \caption{Evolution of $v_3 = \max\{ \underline{v}, v_i - \Delta_v\}$.}
    \label{fig:v3_lb}
\end{subfigure}
\begin{subfigure}[b]{0.33\textwidth}
    \includegraphics[scale=0.4]{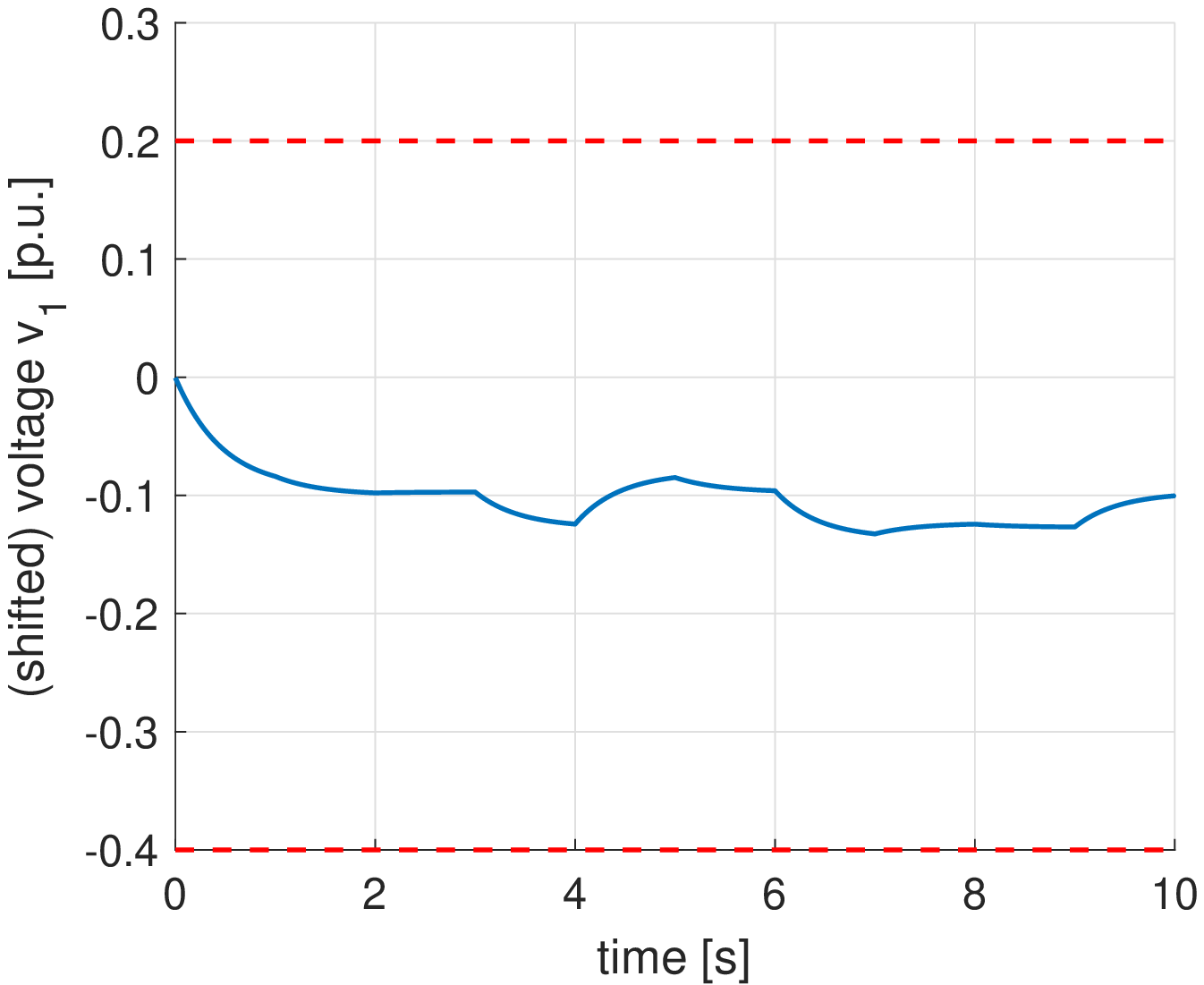}
    \caption{Evolution of $v_1$, kept in $S_v$ by $u_1^q$.}
    \label{fig:v1_lb}
\end{subfigure}
\begin{subfigure}[b]{0.32\textwidth}
    \includegraphics[scale=0.4]{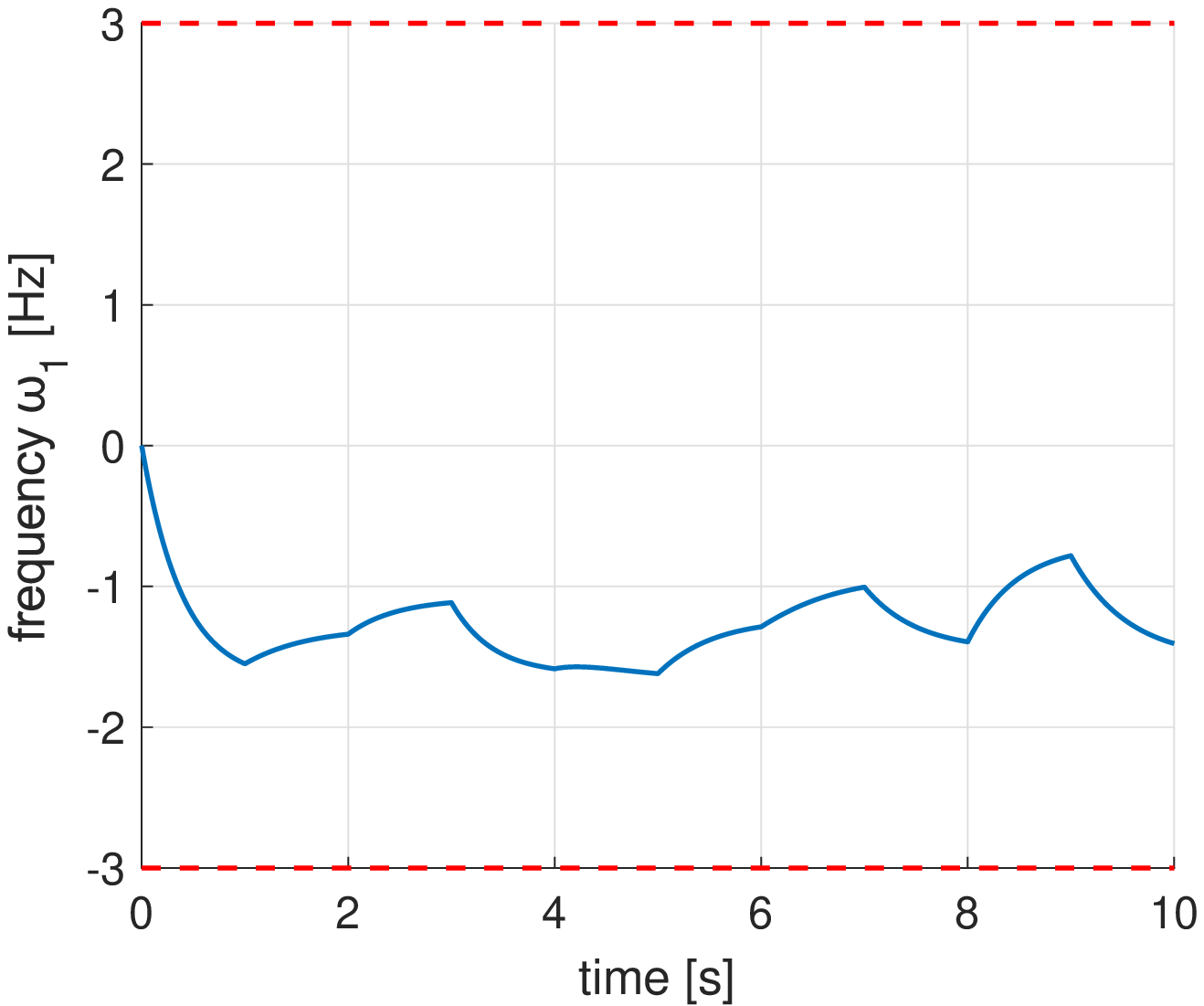}
    \caption{Evolution of $\omega_1$, kept in $S_\omega$ by $u_1^p$.}
    \label{fig:w1_lb}
\end{subfigure}
\caption{Simulation of the voltage $v_1$ and frequency $\omega_1$ under stochastic safety admissible controls $u_1^q$ and $u_1^p$, and lower bound choice of neighbor states $\theta_2$, $\theta_3$, $v_2$ and $v_3$.}
\label{fig:simulation_lb}
\end{figure*}

\section{Conclusion and Future Work}\label{sec:conclusion}

In this paper we considered the problem of transient safety in inverter-based microgrids. Relying on Nagumo's theorem, we developed two approaches to enforce the invariance of frequency and voltage sets of droop-controlled inverters. We solved the resulting optimization problems  with SOS algorithms and successfully illustrated the safety methods on a microgrid model.

There are three promising avenues of future work. 
We first want to compare the efficiency of our approach in terms of size of invariant set and of computation times with barrier function and explicit governor approaches.
We believe that a similar method can be used to handle safe energy storage, with only adding a state of charge constraint to the problem.
Finally, we want to demonstrate our approach in conjunction with system level optimal dispatch problem.

\section*{Acknowledgment}
This research was supported by the Resilience through Data-driven Intelligently-Designed Control (RD2C) Initiative, under the Laboratory Directed Research and Development (LDRD) Program at Pacific Northwest National Laboratory (PNNL). PNNL is a multi-program national laboratory operated for the U.S. Department of Energy (DOE) by Battelle Memorial Institute under Contract No. DE-AC05-76RL01830.

\bibliographystyle{IEEEtran}
\bibliography{references.bib}

\end{document}